\documentclass[12pt]{article}
\usepackage{amsmath}
\usepackage{amssymb}
\usepackage{amsthm}
\usepackage{tikz}
\usetikzlibrary{decorations.markings,decorations.pathreplacing,calc}
\usetikzlibrary{arrows.meta}
\usepackage{algorithm}
\usepackage{algpseudocode}
\usepackage{thm-restate}
\usepackage{verbatim}
\usepackage[framemethod=tikz]{mdframed}
\usepackage{complexity}

\usepackage[margin=1in]{geometry}
\usepackage{subfig}
\usepackage[utf8]{inputenc}
\usepackage{soul}
\usepackage[colorlinks = true]{hyperref}

\usepackage{xcolor}
\definecolor{darkred}  {rgb}{0.5,0,0}
\definecolor{darkblue} {rgb}{0,0,0.5}
\definecolor{darkgreen}{rgb}{0,0.5,0}

\hypersetup{
  urlcolor   = blue,         
  linkcolor  = darkblue,     
  citecolor  = darkgreen,    
  filecolor  = darkred       
}

\newcommand{\be}{\begin{equation}}
\newcommand{\ee}{\end{equation}}
\newcommand{\ba}{\begin{array}}
\newcommand{\ea}{\end{array}}
\newcommand{\bea}{\begin{eqnarray}}
\newcommand{\eea}{\end{eqnarray}}

\newcommand{\br}[1]{\left(#1\right)}

\newtheorem{lemma}{Lemma}

\newtheorem{theorem}{Theorem}
\newtheorem*{theorem*}{Theorem}

\newcommand{\footremember}[2]{%
    \footnote{#2}
    \newcounter{#1}
    \setcounter{#1}{\value{footnote}}%
}
\newcommand{\footrecall}[1]{%
    \footnotemark[\value{#1}]%
}

\title{Beyond product state approximations for a quantum analogue of Max Cut } 
\author{  Anurag Anshu \footremember{iqc}{Institute for Quantum Computing, University of Waterloo, Canada} \footremember{co}{Department of Combinatorics and Optimization, University of Waterloo, Canada}\footremember{pi}{Perimeter Institute for Theoretical Physics, Canada}%
   \and David Gosset \footrecall{iqc} \footrecall{co}%
	\and Karen Morenz \footrecall{iqc} \footrecall{co} \footremember{UT}{Department of Chemistry, University of Toronto, Canada}
  }
\date{}

\begin{document}
\maketitle

\begin{abstract}
We consider a computational problem where the goal is to approximate the maximum eigenvalue of a two-local Hamiltonian that describes Heisenberg interactions between qubits located at the vertices of a graph. Previous work has shed light on this problem's approximability by \textit{product states}. For any instance of this problem the maximum energy attained by a product state is lower bounded by the Max Cut of the graph and upper bounded by the standard Goemans-Williamson semidefinite programming relaxation of it.  Gharibian and Parekh described an efficient classical approximation algorithm for this problem which outputs a product state with energy at least $0.498$ times the maximum eigenvalue in the worst case, and observe that there exist instances where the best product state has energy $1/2$ of optimal. We investigate approximation algorithms with performance exceeding this limitation which are based on optimizing over tensor products of few-qubit states and shallow quantum circuits. We provide an efficient classical algorithm which achieves an approximation ratio of at least $0.53$ in the worst case. We also show that for any instance defined by a $3$ or $4$-regular graph, there is an efficiently computable shallow quantum circuit that prepares a state with energy larger than the best product state (larger even than its semidefinite programming relaxation).

\end{abstract}

\section{Introduction}
In this paper we continue a line of recent work which aims to understand the power and limitations of approximation algorithms for quantum constraint satisfaction problems. Consider an $n$-qubit local Hamiltonian of the form
\begin{equation}
H=\sum_{ij} h_{ij}.
\label{eq:ham}
\end{equation}
Here each term $h_{ij}$ is a Hermitian operator which acts nontrivially only on qubits $i$ and $j$ and we shall assume that $h_{ij}\geq 0$. Estimating the maximum energy $\|H\|$ is a quantum constraint satisfaction problem which is a special case of the well-studied $2$-local Hamiltonian problem, and it is known that computing an estimate of $\|H\|$ within a given small additive error $\epsilon=1/\mathrm{poly}(n)$ is QMA-complete \cite{kitaev2002classical, kempe2004complexity}. Consequently, this sort of precise estimate is unlikely to admit efficient algorithms. 
An estimate $\lambda$ is an $r$-approximation of $\|H\|$, or achieves approximation ratio $r$, if
\[
r \leq \frac{\lambda}{\|H\|}\leq 1.
\]
The classical PCP theorem places stringent bounds on the efficiency of good approximation algorithms for this problem even in the special case where $H$ is diagonal in the computational basis. It states that there exists a constant $r<1$ such that  computing an $r$-approximation to $\|H\|$ is NP-hard \cite{arora1998proof}.   A major open question in this area is whether or not the problem is in fact QMA-hard for some $r<1$. Whereas the standard PCP theorem already implies hardness of approximation, the quantum PCP conjecture targets the more fine-grained question of whether or not such approximations can be checked efficiently given a concise classical witness. These considerations also motivate the study of efficient classical or quantum algorithms for such quantum approximation problems, as measured by the achievable approximation ratio.

A natural way to establish a lower bound $\|H\|\geq \alpha$ is to exhibit a state $|\phi\rangle$ satisfying $\langle \phi|H|\phi\rangle \geq \alpha$.  Several previous works have bounded the approximation ratios that can be achieved by product states $\phi=\phi_1\otimes \phi_2\otimes \ldots \phi_n$ \cite{gharibian2012approximation, brandao2016product, harrow2017extremal, bravyi2019approximation, gharibian2019almost}. Gharibian and Kempe have shown that there always exists a product state which achieves an approximation ratio $r=0.5$ \cite{gharibian2012approximation}. This is also easily seen to be the best possible approximation guarantee for product states, as there are simple examples which saturate this bound. It is not known if a product state achieving a ratio $1/2$ can be computed efficiently in the general case; the most recent progress is an efficient algorithm which outputs a product state that achieves a ratio of $r=0.328$ \cite{qip2020report}. On the other hand, it is known that efficient classical algorithms can achieve approximation ratios arbitrarily close to $1$ if we are willing to specialize to certain families of $2$-local Hamiltonians. Such algorithms are known if the graph which describes the nonzero interactions between qubits is either (a) a $d=O(1)$ dimensional lattice, (b) a planar graph \cite{bansal2007classical, brandao2016product} or (c) dense graphs, in which the number of edges is close to maximal, i.e. $\Omega(n^2)$ \cite{gharibian2012approximation, brandao2016product}.

For completeness, we note that Ref. \cite{bravyi2019approximation} considers a different approximation problem for Hamiltonians where the terms $h_{ij}$ are traceless (rather than positive semidefinite) and describes an efficient $r=O(1/\log(n))$ approximation algorithm based on product states, generalizing the classical result of Charikar and Wirth \cite{charikar2004maximizing}. A related work  \cite{harrow2017extremal} considers a slightly different notion of approximation ratio, again achieved by product states in the traceless setting.

An $n$-qubit product state is an appealing generalization of a classical $n$-bit string, and has the desirable feature that it can be manipulated and stored efficiently by classical algorithms. Moreover, some of the known approximation algorithms for classical constraint satisfaction problems which are based on semidefinite programming have a natural extension to product states \cite{bravyi2019approximation, gharibian2019almost}. But how good are these algorithms, and can we hope to do better using efficient algorithms that are based on entangled states? Of course, most $n$-qubit states do not have concise classical descriptions and cannot even be prepared efficiently using a quantum computer. Since we are aiming for efficient algorithms, we shall restrict our attention to entangled quantum states prepared by polynomial size quantum circuits. 

We shall focus our attention on a specific family of Hamiltonians studied previously in Ref. \cite{gharibian2019almost} which defines a quantum analogue of the Max Cut problem. Unless otherwise specified, throughout this paper we shall consider graphs $G = (V,E,w)$ with nonnegative edge weights $w: E\rightarrow \mathbb{R}_{\geq 0}$, and we write $n=|V|$. We shall also assume that the maximum edge weight is upper bounded by $O(n^{c})$ for some $c=O(1)$.  

For completeness we begin by reviewing facts about the classical Max Cut problem. Recall that the maximum cut of a weighted graph $G$ is defined to be 
\begin{equation}
\mathrm{MC}(G)=\max_{z\in \{\pm 1\}^n}  \mathrm{Cut}_G(z) \qquad \text{where} \qquad \mathrm{Cut}_G(z)=\sum_{\{i,j\} \in E} \frac{w_{ij}}{2} ( I - z_iz_j).
\label{eq:mc}
\end{equation}
An approximation algorithm for the Max Cut problem due to Goemans and Williamson \cite{goemans1995improved} is based on the following semidefinite programming relaxation of Eq.~\eqref{eq:mc}:
\begin{equation}
\mathrm{SDP}(G)=\max_{M\in \mathbb{R}^{n\times n}: M\geq 0, \mathrm{diag}(M)=I} \sum_{\{i,j\} \in E} \frac{w_{ij}}{2} ( I - M_{ij}).
\label{eq:sdp}
\end{equation}

A matrix $M$ achieving the maximum SDP value $\mathrm{SDP}(G)$ can be computed efficiently using standard classical algorithms. The Goemans-Williamson algorithm then uses a randomized procedure which maps $M$ to a bit string $z$ which is guaranteed to satisfy
\begin{equation}
\mathrm{Cut}_G(z)\geq 0.8785 \cdot \mathrm{SDP}(G)
\label{eq:gw}
\end{equation}
for all graphs $G$  \cite{goemans1995improved}.

The quantum Max Cut problem as considered in Ref. \cite{gharibian2019almost} is defined by a family of local Hamiltonians Eq.~\eqref{eq:ham} where each term $h_{ij}$ is proportional to the two-qubit singlet state $|s\rangle=\sqrt{2}^{-1}(|01\rangle-|10\rangle)$. In particular, given a graph $G=(V,E,w)$ we define
\begin{equation} 
H_G = \sum_{\{i,j\} \in E} w_{ij} h_{ij} \qquad \quad h_{ij} = \frac{1}{2}\left(I - X_iX_j - Y_iY_j - Z_iZ_j\right)=2|s\rangle\langle s|_{ij}.
\label{eq:hg}
\end{equation}

We are interested in approximating the maximum eigenvalue of $H_G$ which we write as
\[
\mathrm{OPT}(G)=\|H_G\|.
\]
Estimating this quantity can be viewed as a quantum analogue of the classical Max Cut problem. Indeed, a constraint $(I-z_iz_j)$ in the Max Cut problem Eq.~\eqref{eq:mc} has maximal energy when the corresponding two entries disagree, i.e., $z_i\neq z_j$. Analogously, a constraint $h_{ij}$ in the Hamiltonian Eq.~\eqref{eq:hg} has maximal energy for a quantum state $|\psi\rangle$ when the two qubits are antisymmetric under swap, i.e., $SWAP_{ij}|\psi\rangle=-|\psi\rangle$. In this sense, the classical and quantum constraints represent two different notions of disagreement. 

Piddock and Montanaro have shown that the problem of computing a precise estimate of $\mathrm{OPT}(G)$ is QMA-complete \cite{piddock2015complexity}, and recent work has focused on its approximability using product states \cite{gharibian2019almost}. Let us now see how the problem of optimizing the energy of Eq.~\eqref{eq:hg} over product states is directly related to the Max Cut problem Eq.~\eqref{eq:mc} and its semidefinite relaxation Eq.~\eqref{eq:sdp}. An $n$-qubit product state $\phi$ can be specified (up to a global phase) by $n$ normalized vectors $v^{(j)}\in \mathbb{R}^3$:
\[
|\phi\rangle\langle \phi|=\bigotimes_{j=1}^{n} \frac{1}{2}\left(I+v^{(j)}_1 X+v^{(j)}_2 Y+v^{(j)}_3 Z\right) \qquad \qquad \|v^{(j)}\|=1,
\]
and its energy is given by
\begin{equation}
\mathrm{Tr}\left[|\phi\rangle\langle\phi| H_G\right]=  \sum_{\{i,j\} \in E} \frac{w_{ij}}{2} (1-v^{(i)}\cdot v^{(j)}).
\end{equation}
Defining 
\begin{equation}
\alpha(k)=\max_{\{v_i\in \mathbb{R}^k: \|v_i\|=1\}} \sum_{\{i,j\} \in E} \frac{w_{ij}}{2} (1-v_i\cdot v_j).
\end{equation}
we see that 
\[
\alpha(1)=\mathrm{MC}(G)\leq \mathrm{PROD}(G)=\alpha(3)\leq \mathrm{SDP}(G)=\alpha(n).
\]
The Goemans-Williamson algorithm for the Max Cut problem has been generalized by Briet, de Oliveira Filho, and Vallentin to obtain efficient algorithms for approximating $\alpha(k)$ for $1<k<n$ \cite{briet2010positive}. The resulting approximation ratios obtained become larger as $k$ increases towards $k=n$ where the optimal value can be computed efficiently and exactly by semidefinite programming. Their result for the case $k=3$ at hand is summarized below.
\begin{theorem}[\cite{briet2010positive}]
There exists an efficient randomized classical algorithm which computes an estimate $\mu$ such that
\[
0.956 \leq \frac{\mu}{\mathrm{PROD}(G)}\leq 1.
\]
\label{thm:prodapprox}
\end{theorem}
This algorithm (and other randomized algorithms discussed in this paper) may fail with some small probability, say $0.01$, in which case the output of the algorithm is a flag indicating failure.

Since $\mathrm{PROD}(G)\geq 0.5 \cdot \mathrm{OPT}(G)$ \cite{gharibian2012approximation}, the algorithm described in Theorem \ref{thm:prodapprox} can be used to approximate $\mathrm{OPT}(G)$ with ratio at least $0.5\cdot 0.956=0.478$, as observed in Ref.~\cite{GharibianLiu20}. The recently proposed approximation algorithm of Gharibian and Parekh  \cite{gharibian2019almost} is based on rounding a solution to a semidefinite program relaxation of $\mathrm{OPT}(G)$, and obtains an even higher ratio of $0.498$. The authors of Ref. \cite{gharibian2019almost} note that their algorithm is almost optimal (as far as product states are concerned), since there exists a very simple graph---just two vertices connected by a single weight one edge--- for which the optimal product state is equal to $0.5\cdot \mathrm{OPT}(G)$.

Our first result shows that if all edge weights are equal then this limitation of product states only occurs in small graphs. That is, for sufficiently large connected graphs with uniform weights, it is always possible to efficiently find a product state with a strictly larger approximation ratio:
\begin{theorem}
Suppose $G=(V,E,w)$ is a connected and unweighted graph, i.e., $w_{ij}=1$ for all $\{i,j\}\in E$. Then 
\begin{equation}
\frac{\mathrm{PROD}(G)}{\mathrm{OPT}(G)}\geq  \frac{4}{7}-O(|E|^{-1}).
\end{equation}
The efficient randomized algorithm from Theorem \ref{thm:prodapprox} computes an $r$-approximation to $\mathrm{OPT}(G)$, where $r\geq 0.546-O(|E|^{-1})$.
\label{thm:prod1}
\end{theorem}

In the more general setting where the weights may not be uniform, one can of course construct examples of connected graphs where all weights are vanishingly small except for the weight of a single edge. In this limit we already know that it is impossible to beat a ratio of $0.5$ using product states. Our next result shows that by considering tensor products of one- and two-qubit states it is possible to guarantee a strictly better approximation ratio.
\begin{theorem}
Let $G=(V,E,w)$ be a weighted graph. Then there is a tensor product $\phi=\phi_1\otimes \phi_2\otimes\ldots \phi_L$ of $1$- and $2$-qubit states $\{\phi_j\}$ such that
\[
\frac{\langle \phi|H_G|\phi\rangle} {\mathrm{OPT}(G)}\geq 0.55.
\]
Moreover, there is an efficient randomized algorithm which outputs an $r$-approximation to $\mathrm{OPT}(G)$, where $r\geq 0.53$.
\label{thm:prod2}
\end{theorem}

Theorem \ref{thm:prod2} provides the best currently known efficient approximation algorithm for this problem, improving slightly on Ref. \cite{gharibian2019almost}. Moreover, it establishes that although there exist graphs where the best product state is only $1/2$ of the optimal energy, efficient classical algorithms can go slightly beyond this ratio. 

Our next result shows that, for a family of low-degree graphs it is possible to efficiently beat product states on \textit{every} graph from the family. In particular, given any $3$- or $4$-regular graph $G$, we can efficiently compute a constant-depth quantum circuit which prepares a state with energy strictly larger than the best product state energy $\mathrm{PROD}(G)$ (in fact, larger than its semidefinite relaxation $\mathrm{SDP}(G)$).

\begin{theorem}
Suppose $G=(V,E,w)$ is a $k$-regular graph with $k\in \{3,4\}$. There is a depth-$(k+1)$ quantum circuit $U(G)$ that can be efficiently computed by a randomized classical algorithm such that the state $|\phi\rangle=U(G)|0^n\rangle$ approximates $\mathrm{OPT}(G)$ with a strictly larger ratio than that of any product state. Moreover,
\[
\frac{\langle \phi| H_G |\phi\rangle}{\mathrm{PROD}(G)}\geq \frac{\langle \phi| H_G |\phi\rangle}{\mathrm{SDP}(G)} >1.001
\] 
\label{thm:lowdegree}
\end{theorem}

The low depth quantum circuit used in Theorem \ref{thm:lowdegree} is inspired (and similar to) the quantum approximate optimization algorithm described by Farhi, Goldstone, and Gutmann \cite{farhi2014quantum}. The circuit is directly obtained from any computational basis state $z\in \{0,1\}^n$ with a large enough cut value $\mathrm{Cut}_G(z)$\footnote{Note that we previously defined $\mathrm{Cut}_G(z)$ for inputs $z\in \{-1,1\}^n$.  Here and below we extend this definition to bit string inputs $z\in \{0,1\}^n$ by identifying each bit $z_j$ with the $\pm 1$-valued variable $(-1)^{z_j}$.}; in particular, it is sufficient to use a bit string satisfying Eq.~\eqref{eq:gw} which can be computed efficiently using the Goemans-Williamson algorithm. The quantum computation starts from the computational basis state $|z\rangle$ and then applies a low-depth quantum circuit composed of a sequence of commuting two-qubit gates of the form 
\[
e^{i\theta P(j) P(k)}
\]
where for each qubit $v$ we choose a Pauli operator $P(v)\in \{X_v, Y_v\}$ depending only on the bit $z_v$. To prove the theorem we compute the energy of this state as a function of the variational parameter $\theta$ and then optimize. 

In summary, we have shown that for the quantum Max Cut problem there are efficient algorithms which beat any approximation algorithm based on product states. A natural open question is whether this is also true for the more general problem of approximating the maximum energy of a two-local Hamiltonian Eq.~\eqref{eq:ham}. One may also ask if the semidefinite programming method \cite{goemans1995improved} can be used in some novel way to efficiently obtain approximation ratios which go beyond the limitations of product states. For the quantum Max Cut problem, Ref. \cite{gharibian2019almost} provides a semidefinite program which upper bounds the optimal energy $\mathrm{OPT}(G)$. A central challenge here is that we do not (yet) know a randomized rounding scheme which maps an SDP solution to an entangled state.

\section{Tensor products of few qubit states}
In this Section, we prove Theorems \ref{thm:prod1} and \ref{thm:prod2}. We shall use the following upper bound for the special case where $G$ is a star graph. The lemma shows that the maximum energy for any star with at least $3$ vertices is always less than the trivial upper bound $2\sum_{e\in E}{w_e}$ which comes from the triangle inequality. This can be interpreted as a consequence of the monogamy of entanglement--the center spin cannot be maximally entangled with all of the points of the star. Along similar lines,  Ref.~\cite{GharibianLiu20} provides a different upper bound on $\|H_G\|$ using a monogamy of entanglement bound known as the Coffman-Kundu-Wooters inequality.
\begin{lemma}
Suppose $G=(V,E,w)$ is a star graph with nonnegative weights. Then 
\begin{equation}
\|H_G\|\leq \max_{e\in E} w_e+ \sum_{e\in E}w_e.
\label{eq:lam}
\end{equation}
\label{lem:star}
\end{lemma}

\begin{proof}
Define the total spin operators
\[
\vec{S}=\frac{1}{2}\left(\sum_{j\in V} X_j, \sum_{j\in V} Y_j, \sum_{j\in V} Z_j\right).
\]
Let $S_x=\frac{1}{2}\sum_{j\in V} X_j$, $S_y= \frac{1}{2}\sum_{j\in V} Y_j$, $S_z= \frac{1}{2}\sum_{j\in V} Z_j$ and note that the Hamiltonian Eq.~\eqref{eq:hg}, $S^2=S_x^2+S_y^2+S_z^2$, and $S_z$ are mutually commuting. It is shown in Ref. \cite{lieb1962ordering} that the maximum eigenvalue of Eq.~\eqref{eq:hg} on any (nonnegatively) weighted complete bipartite graph with bipartition $V=A\sqcup B$ is attained by an eigenvector $\phi$  which satisfies
\[
S^2 |\phi\rangle = s(s+1)|\phi\rangle \qquad \quad S_z|\phi\rangle=s|\phi\rangle \qquad \qquad s=\left(|A|-|B|\right)/2.
\]
A star graph is a complete bipartite graph with $|A|=|V|-1$ and $|B|=1$. Therefore the result of Lieb and Mattis implies that a maximum eigenvalue is attained by a state $\phi$ satisfying $S_z|\phi\rangle=|V|/2-1$. In particular, $\phi$ is equal to the maximum eigenvector of $H_{G}$ restricted to the $|V|$-dimensional subspace 
\[
\mathcal{Q}=\mathrm{span}\{|100\ldots0\rangle,|010\ldots 0\rangle,\ldots, |00\ldots,01\rangle\}
\]
spanned by computational basis states with Hamming weight equal to $1$. It is easily seen that the Hamiltonian restricted to this subspace is the Laplacian matrix of $G$. More precisely,
\[
H_{G}|_{\mathcal{Q}}=L(G)
\]
where $L(G)$ is the graph Laplacian of $G$,  defined by 
\[
L(G)_{ij}=\begin{cases} \sum_{e\sim i} w_e &  i=j\\ -w_{e} & e=\{i,j\}\in E\\ 0 & otherwise.\end{cases}
\]
The lemma follows by upper bounding the norm of the Laplacian of a star graph
\begin{equation}
\|L(G)\|\leq \sum_{e\in E} w_e+\max_{e\in E} w_e.
\label{eq:laplaciannorm}
\end{equation}
The upper bound Eq.~\eqref{eq:laplaciannorm} is obtained using an argument from Ref. \cite{merris1998note}. In particular, we note that
\[
\|L(G)\|=\|W^{-1} L(G)W\| 
\]
where $W$ is a diagonal matrix such that $W_{ii}=\sum_{e\sim i} w_e$, and then use Gershgorin's circle theorem to upper bound the right hand side. Computing the Gershgorin discs for a star graph we arrive at
\begin{align}
\|W^{-1} L(G)W\|& \leq \max \left\{\sum_{e\in E} w_e+\left(\sum_{e\in E} w_e\right)^{-1}\sum_{e\in E} w_e^2 , \quad \max_{e\in E} w_e+\sum_{e\in E}w_e\right\}\\
&=\max_{e\in E} w_e+\sum_{e\in E}w_e.
\end{align}
\end{proof}
We note that for a star graph with uniform weights the upper bound Eq.~\eqref{eq:lam} becomes an equality, as can be seen using the rules for addition of angular momentum. 

Next, we consider the case of uniform weights $w_{ij}=1$ on an arbitrary connected graph. Using Lemma \ref{lem:star}, we exhibit a product state with approximation ratio better than $\frac{1}{2}$.

\begin{theorem}
Suppose $G=(V,E,w)$ is a connected graph with uniform weights, i.e., $w_{ij}=1$ for all $\{i,j\}\in E$. Then
\begin{equation}
\frac{\mathrm{PROD}(G)}{\mathrm{OPT}(G)}\geq \frac{1}{3}+\frac{2}{3}\left(\frac{|E|}{2|E|+|V|}\right).
\label{eq:statement}
\end{equation}
Moreover, there exists a computational basis state with energy satisfying the above inequality.
\label{thm:dense}
\end{theorem}

\begin{proof}

For any vertex $v\in V$ define a Hamiltonian $h_v$ which has support only on qubit $v$ and its neighbors:
\[
h_v=\sum_{j: \{v,j\}\in E} \frac{1}{2}\left(I- X_v X_j -Y_v Y_j-Z_v Z_j\right).
\]
Note that we may write $H_G=\frac{1}{2}\sum_{v\in V} h_v$, where the factor of $1/2$ compensates for the fact that the Hamiltonian term corresponding to each edge of the graph appears twice on the right hand side. Now using the triangle inequality we get
\begin{equation}
\mathrm{OPT}(G)\leq \frac{1}{2}\sum_{v\in V} \|h_v\|.
\label{eq:triangle}
\end{equation}
Let us write $d_v$ for the degree of vertex $v$. Then
\[
\|h_v\|\leq d_v+1,
\]
where we used Lemma \ref{lem:star}. Substituting in Eq.~\eqref{eq:triangle} gives
\begin{equation}
\mathrm{OPT}(G)\leq \frac{1}{2}\sum_{v\in V} (d_v+1)=|E|+|V|/2.
\label{eq:upp1}
\end{equation}
To see why Eq.~\eqref{eq:upp1} is nontrivial, note that since $G$ is a connected graph on $|V|$ vertices, it satisfies $|E|\geq |V|-1$ (the minimum is attained by a tree). Thus Eq.~\eqref{eq:upp1} implies
\begin{equation}
\mathrm{OPT}(G)\leq \frac{1}{2}(3|E|+1),
\label{eq:upp}
\end{equation}
which improves upon the naive upper bound $\mathrm{OPT}(G)\leq 2|E|$ which is obtained by applying the triangle inequality directly to Eq.~\eqref{eq:hg}. 

We need only a little bit more to get the Theorem from Eq.~\eqref{eq:upp1}.  Let us write
\[
H_G=\frac{|E|}{2}+H^{X}(G)+H^{Y}(G)+H^{Z}(G)
\]
where
\begin{equation}
H^{X}(G)=-\frac{1}{2}\sum_{\{i,j\}\in E} X_iX_j \qquad H^{Y}(G)=-\frac{1}{2}\sum_{\{i,j\}\in E} Y_iY_j\qquad H^{Z}(G)=-\frac{1}{2}\sum_{\{i,j\}\in E} Z_iZ_j.
\end{equation}
We denote their largest eigenvalues as $\lambda_{\mathrm{max}}^{P}(G)$ with $P=X,Y,Z$. Note that these 3 quantities are all equal. Applying the triangle inequality and using this fact gives
\begin{equation}
\mathrm{OPT}(G)\leq \frac{|E|}{2}+3\lambda_{\mathrm{max}}^{Z}(G).
\label{eq:maxbnd}
\end{equation}
Also note that we can lower bound $\mathrm{PROD}(G)$ by the maximum energy of a computational basis state:
\begin{equation}
\mathrm{PROD}(G)\geq \frac{|E|}{2}+\lambda_{\mathrm{max}}^{Z}(G).
\label{eq:prodbnd}
\end{equation}

Now combining Eqs.~(\ref{eq:maxbnd}, \ref{eq:prodbnd}) gives
\[
\mathrm{PROD}(G)\geq \frac{|E|}{2}+\frac{\mathrm{OPT}(G)-|E|/2}{3}=\frac{1}{3}\mathrm{OPT}(G)+\frac{1}{3}|E|.
\]
Therefore
\[
\frac{\mathrm{PROD}(G)}{\mathrm{OPT}(G)}\geq \frac{1}{3}+\frac{1}{3} \frac{|E|}{\mathrm{OPT}(G)}.
\]
Finally, substituting Eq.~\eqref{eq:upp1} in the second term we arrive at Eq.~\eqref{eq:statement} and complete the proof.
\end{proof}

\subsection*{Proof of Theorem \ref{thm:prod1}}
\begin{proof}
Let $T$ be a spanning tree of $G$, which can be computed efficiently and has $|V|-1$ edges. Let $s\in \{0,1\}^n$ be a bit string corresponding to a 2-coloring of $T$, i.e., $s_i\neq s_j$ whenever $\{i,j\}$ is an edge of $T$ (of course, $s$ can also be computed efficiently). Then 
\[
\langle s|H_G|s\rangle=\mathrm{Cut}_G(s)\geq |V|-1,
\]
and combining with Eq.~\eqref{eq:upp1} gives
\[
\frac{\langle s|H_G|s\rangle}{\mathrm{OPT}(G)}\geq \frac{2|V|-2}{2|E|+|V|}.
\]
Putting this together with Theorem \ref{thm:dense} we arrive at
\begin{align}
\frac{\mathrm{PROD}(G)}{\mathrm{OPT}(G)}& \geq \max\left\{\frac{2(|V|-1)}{2|E|+|V|}, \frac{4|E|+|V|}{6|E|+3|V|}\right\}\\
\end{align}
Now let $x=(|V|-1)/|E|$ and note that $x\in [0,1]$, and 
\begin{align}
\frac{\mathrm{PROD}(G)}{\mathrm{OPT}(G)}& \geq \min_{0\leq x\leq 1} \max\left\{\frac{2x}{2+x}, \frac{4+x}{6+3x}\right\}-O(|E|^{-1})\\
&= 4/7-O(|E|^{-1}).
\label{eq:finalratprod}
\end{align}
The randomized approximation algorithm of \ref{thm:prodapprox} outputs an estimate which is an $\alpha$-approximation to $\mathrm{PROD}(G)$ with ratio $\alpha\geq 0.956$. Eq.~\eqref{eq:finalratprod} implies that this estimate is an $r$-approximation of $\mathrm{OPT}(G)$ with $r\geq \alpha\cdot (4/7-O(|E|^{-1}))=0.546-O(|E|^{-1}))$.
\end{proof}

\subsection*{Proof of Theorem \ref{thm:prod2}}
\begin{proof}
Note that in the weighted case we may run through exactly the same arguments used to obtain Eq.~\eqref{eq:statement}.  Eq.~\eqref{eq:upp1} is replaced by
\[
\mathrm{OPT}(G)\leq W+\frac{1}{2}\sum_{v\in V} \max_{e\sim v} w_e 
\]
where $W=\sum_{e\in E}w_e$, and correspondingly we have
\begin{equation}
\frac{\mathrm{PROD}(G)}{\mathrm{OPT}(G)}\geq \frac{1}{3}+\frac{2}{3}\left(\frac{W}{2W+\sum_{v\in V} \max_{e\sim v} w_e}\right).
\label{eq:compenergy}
\end{equation}

Now let us focus on the expression
\[
\sum_{v\in V} \max_{e\sim v} w_e.
\]
We note that this quantity can be trivially upper bounded as $2W$ since each edge can appear at most twice in the sum (once for each of its incident vertices). This naive upper bound is not sufficient for our purposes, and so we perform a more careful analysis below.
\begin{lemma}
We may efficiently compute edge subsets $M,F\subseteq E$ such that $M$ is a matching and $F$ is a forest, and
\[
\sum_{v\in V} \max_{e\sim v} w_e= \sum_{e\in M} w_e+\sum_{e\in F} w_e.
\]
\label{lem:matchfor}
\end{lemma}
\begin{proof}
Let us fix an ordering $e_1,e_2,\ldots, e_m$ of all the edges of $G$ such that 
\[
w_{e_1}\leq w_{e_2}\leq \ldots \leq w_{e_m}
\]
(if all edge weights are distinct there is a unique such ordering, otherwise there is some freedom in the choice). Now for each vertex $v\in V$ we let $I(v)\in E$ be the (unique) edge incident to $v$ which is maximal with respect to the above ordering. We define
\begin{align*}
F& =\left\{I(v): v\in V\right\}\\
M& =\left\{e\in E: e=I(v) \quad \text{and} \quad e=I(w) \quad \text{for two distinct vertices} \quad v\neq w\in V\right\}.
\end{align*}
At most one edge $e=I(v)$ incident to any given vertex $v$ can appear in $M$, and hence $M$ is a matching. To see that $F$ is a forest, consider a graph $G'=(V,E,w')$ with the same vertex and edge sets as $G$, but where the edge weights are rescaled so that $w'(e_j)=-j$ for $1\leq j\leq m$ (in particular, all edge weights are negative, distinct, and their magnitudes respect our chosen ordering). Then each edge of $F$ is contained in any minimum spanning tree of $G'$,  by the well-known cut property of minimum spanning trees. We infer that $F$ does not contain any cycles, and is therefore a forest.
\end{proof}

Now let $M,F$ be as in the lemma, and define the set of vertices $U\subseteq V$ which are not incident to an edge in $M$. Consider a random variable
\[
|\phi_z\rangle=\left(\bigotimes_{e=\{i,j\}\in M} \frac{1}{\sqrt{2}}\left(|01\rangle-|10\rangle\right)_{ij}\right)\otimes |z\rangle_{U} 
\]
where $z\in \{0,1\}^{|U|}$ is a uniformly random bit string. Then
\begin{equation}
\mathbb{E}_{z}[ \langle \phi_z| H_{G}|\phi_z\rangle]= 2\sum_{e\in M} w_e+\frac{1}{2}(W-\sum_{e\in M} w_e)=\frac{3}{2}m+\frac{1}{2}W \qquad \qquad m\equiv\sum_{e\in M} w_e.
\label{eq:matchenergy}
\end{equation}
This shows that there exists a state $\phi_z$ with energy at least $\frac{3}{2}m+\frac{1}{2}W$.

Finally, since $F$ is a forest, we may efficiently compute a computational basis state $s\in \{0,1\}^n$ such that
\begin{equation}
\langle s|H_G|s\rangle\geq f \qquad \qquad f\equiv\sum_{e\in F} w_e.
\label{eq:forestenergy}
\end{equation}

This follows since the Max Cut for a forest is achieved by an efficiently computable 2-coloring of the vertices. Putting together Eqs.~(\ref{eq:compenergy},\ref{eq:matchenergy},\ref{eq:forestenergy}) and Lemma \ref{lem:matchfor}, we see that there exists a tensor product $\phi$ of $1$- and $2$-qubit states such that
\begin{align}
\frac{\langle\phi|H_G|\phi\rangle}{\mathrm{OPT}(G)}&\geq \max\left\{\frac{2f}{2W+f+m}, \frac{3m+W}{2W+f+m}, \frac{1}{3}+\frac{2}{3}\left(\frac{W}{2W+f+m}\right)   \right\}\\
&\geq \ \min_{0\leq x\leq y\leq 1} \frac{1}{2+y+x}\max\left\{2y,3x+1, \frac{1}{3}(4+x+y)   \right\}\\
&\geq 0.55,
\end{align}
where in the second line we set $x=m/W, y=f/W$, and in the last line we used a computer.

Now let us bound the approximation ratio achieved by an efficient randomized algorithm. First note that the state $|s\rangle$ in Eq.~\eqref{eq:forestenergy} can be computed efficiently. Moreover, using Eq.~\eqref{eq:matchenergy} and the fact that $\langle \phi_z| H_{G}|\phi_z\rangle$ is a random variable upper bounded by $2W$ we see that the probability of randomly sampling a bit string $z$ with energy at least $\frac{3}{2}m+0.49W$ is 
\[
\mathrm{Pr}\left[\langle \phi_z| H_{G}|\phi_z\rangle\geq \frac{3}{2}m+0.49W\right]\geq \frac{0.01}{1.51}.
\]
By randomly sampling $O(1)$ times, with high probability we will obtain a bit string with this energy. Finally, note that combining Theorem \ref{thm:prodapprox} with Eq.~\eqref{eq:compenergy} we get a randomized algorithm that outputs a state with energy at least
\[
0.956\cdot \left(\frac{1}{3}+\frac{2}{3}\left(\frac{W}{2W+f+m}\right)\right).
\]
Thus we may efficiently compute a state with approximation ratio at least
\begin{align}
\min_{0\leq x\leq y\leq 1} \frac{1}{2+y+x}\max\left\{2y,3x+0.98, \frac{0.956}{3}(4+x+y)   \right\}
\geq 0.53.
\end{align}
\end{proof}

\section{Low degree regular graphs}

In this section we consider the case of $3$- or $4$-regular graphs and we establish Theorem \ref{thm:lowdegree}.

Given an $n$-vertex graph $G=(V,E,w)$, we shall consider the following algorithm. First, we use the classical Goemans-Williamson algorithm \cite{goemans1995improved} to compute a bit string $z\in \{0,1\}^n$ satisfying Eq.~\eqref{eq:gw}. This defines a partition of the edges into those which are satisfied and those which are not:
\begin{equation}
E_{\mathrm{sat}}=\{\{u,v\}\in E: z_u\neq z_v\} \qquad \qquad E_{\mathrm{unsat}}=E\setminus E_{\mathrm{sat}}.
\label{eq:partition}
\end{equation}

Next we define a Pauli operator $P(j)$ for each qubit $1\leq j\leq n$, which depends on the $j$-th bit of $z$:
\[
P(j)={\begin{cases}X_j, & z_j=1\\ Y_j, & z_j=0\end{cases}}.
\]
Finally, we define a variational state
\begin{equation}
|\phi(\theta)\rangle=V(\theta)|z\rangle \quad \text{where} \quad V(\theta)=\exp\left(\sum_{\{j,k\} \in E} i\theta P(j)P(k)\right).
\label{eq:phi}
\end{equation}
Here $\theta\in \mathbb{R}$ is a parameter that we will choose later. Note that $V(\theta)$ can be expressed as a product of \textit{commuting} $2$-qubit gates
\begin{equation}
V(\theta)=\prod_{\{j,k\} \in E} \exp\left(i\theta P(j)P(k)\right).
\label{eq:V}
\end{equation}
Moreover, if the graph $G$ has maximum degree $\Delta$ then we may efficiently compute an edge coloring with $\Delta+1$ colors such that no two edges with the same color share a vertex. If we order the  gates Eq.~\eqref{eq:V} in $\Delta+1$ layers according to this edge coloring we obtain a depth $\Delta+1$ quantum circuit that implements $V(\theta)$. 

\begin{figure}
\centering
\begin{tikzpicture}

\draw[fill=black] (0,0) circle (1pt);
\draw (-.25,0) node {i}; 
\draw (2.25,0) node {j}; 
\draw[fill=black] (2,0) circle (1pt);
\draw[fill=black] (1,-2) circle (1pt);
\draw[fill=black] (1,-1) circle (1pt);
\draw[fill=black] (1,1) circle (1pt);
\draw[fill=black] (1,2) circle (1pt);
\draw (0,0)--(2,0);
\draw (0,0)--(1,1)--(2,0);
\draw (0,0)--(1,-1)--(2,0);
\draw (0,0)--(1,2)--(2,0);
\draw (0,0)--(1,-2)--(2,0);
\end{tikzpicture}
\caption{An edge $\{i,j\}$ contained in 4 triangles.\label{fig:Ttriangles}}
\end{figure}
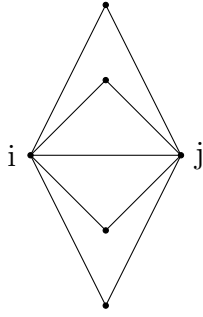

The following lemma describes the energy of the variational state $\phi(\theta)$. Below we write $d_j$ for the degree of vertex $j\in V$. We say that an edge $\{i,j\}\in E$ is contained in $T$ triangles iff there are vertices $k_1,k_2,\ldots, k_T$ such that $\{i,k_1\},\{i,k_2\},\ldots, \{i,k_T\}\in E$ and $\{j,k_1\},\{j,k_2\},\ldots, \{j,k_T\}\in E$. This is depicted in Fig. \ref{fig:Ttriangles}.

\begin{lemma}
Let $G = (V,E,w)$ be a graph and let $\phi(\theta)$ be the variational state defined in Eq.~\eqref{eq:phi}. If $\{i,j\}\in E_{\mathrm{sat}}$ is a satisfied edge contained in exactly $T$ triangles then
\\ \begin{equation}
\langle \phi (\theta) | 2h_{ij} | \phi (\theta) \rangle = 1 + \sin(2\theta)\cos^{d_i-1}(2\theta) + \sin(2\theta)\cos^{d_j-1}(2\theta) + 
\frac{1 + \cos^{T}(4\theta)}{2}\cos^{d_i+d_j-2-2T}(2\theta).\label{eq:sat1}
\end{equation} 
On the other hand, if $\{i,j\}\in E_{\mathrm{unsat}}$ is an unsatisfied edge contained in exactly $T$ triangles, then
\\\begin{align}
\langle \phi (\theta) |2 h_{ij} | \phi (\theta) \rangle& = 1-\cos^{d_i + d_j - 2-2T}(2\theta).\label{eq:unsat1}
\end{align} 
\label{lem:gengraph}
\end{lemma}

We defer the proof of the Lemma until the end of this section. Let us now see how Lemma \ref{lem:gengraph} can be used to lower bound the energy $\langle \phi(\theta)|H_G|\phi(\theta)\rangle$ when $G$ is a $3$- or $4$-regular graph. In fact, we will only need Eq.~\eqref{eq:sat1} for the proof below; Eq.~\eqref{eq:unsat1} is included only for completeness.

\begin{proof}[Proof of Theorem \ref{thm:lowdegree}]

In a $d$-regular graph, each edge may be contained in $T\leq d-1$ triangles. Note that the energy Eq.~\eqref{eq:sat1} of a satisfied edge is lower bounded by the same expression with $T=0$ since the last term is monotonically increasing with $T$. Thus all satisfied edges in a $d$-regular graph have energy lower bounded as

\begin{equation}
\langle \phi (\theta) | 2h_{ij} | \phi (\theta) \rangle \geq  1 + 2\cos^{d-1}(2\theta)\sin(2\theta) + \cos^{2d-2}(2\theta)
\label{eq:satdreg}
\end{equation} 

An unsatisfied edge in a $d$-regular graph always contributes a nonnegative energy since the Hamiltonian terms $h_{ij}$ are positive semidefinite and the weights $w_{ij}$ are nonnegative.

Thus for a $d$-regular graph $G$ we have
\begin{equation}
\langle \phi (\theta) | H_G | \phi (\theta) \rangle \geq \frac{F(\theta,d)}{2}\sum_{\{i,j\}\in E_{\mathrm{sat}}}w_{ij}=\frac{F(\theta,d)}{2}\mathrm{Cut}_G(z),
\label{eq:dreg}
\end{equation} 
where
\[
F(\theta,d)=1+2\cos^{d-1}(2\theta)\sin(2\theta)+\cos^{2d-2}(2\theta).
\]
For a fixed $d$ we may compute $\theta^{\star}(d)=\mathrm{argmax}_{\theta} F(\theta,d)$ which maximizes the right hand side. Also note that since $z$ is the output of the Goemans-Williamson approximation algorithm, it satisfies Eq.~\eqref{eq:gw} and therefore 
\begin{equation}
\frac{\langle \phi (\theta) | H_G | \phi (\theta)\rangle}{\mathrm{SDP}(G)}\geq G(d)\equiv (0.8785)\cdot \frac{F(\theta^{\star}(d),d)}{2}.
\label{eq:defG}
\end{equation}
Using a computer we find $G(3)=1.047\ldots$ and $G(4)=1.001\ldots$, which completes the proof.
\end{proof}
\begin{proof}[Proof of Lemma \ref{lem:gengraph}]
We shall compute
\begin{equation}
\langle \phi | 2h_{ij} | \phi \rangle = 1 - \langle \phi | X_iX_j | \phi \rangle - \langle \phi | Y_iY_j | \phi \rangle - \langle \phi | Z_iZ_j | \phi \rangle
\label{eq:3terms}
\end{equation}
(here and below we write $\phi\equiv \phi(\theta)$ for ease of notation).

We treat the two cases separately: satisfied edges $\{i,j\}\in E_{\mathrm{sat}}$ (i.e., $z_i \neq z_j$) and unsatisfied edges $\{i,j\}\in E_{\mathrm{unsat}}$ (i.e., $z_i=z_j$).
 
\paragraph{Satisfied edge:} Without loss of generality, assume that $z_i=0$ and $z_j=1$ (else we perform the same calculation with $i$ and $j$ interchanged). Using the standard commutation relations between Pauli operators, and the fact that the set of operators $\{P(k)P(\ell)\}_{(k,\ell)\in E}$ mutually commute we get
\begin{align*}
&\langle \phi | X_iX_j | \phi \rangle \\
&= \langle z |\bigg(\prod_{\{k,\ell\} \in E}\exp(-i\theta P(k) P(\ell))\bigg)X_iX_j\bigg(\prod_{\{k,\ell\} \in E}\exp (i\theta P(k) P(\ell))\bigg)|z \rangle\\
&= \langle z |\bigg(\prod_{k:\{i,k\} \in E}\exp(-2i\theta Y_i P(k))\bigg)X_iX_j|z \rangle\\
&= \langle z | \prod_{k:\{i,k\} \in E} \br{\cos(2\theta) - i\sin(2\theta) Y_i P(k)}X_iX_j|z \rangle\\
&= -i\cos^{d_i-1}(2\theta)\sin(2\theta)\langle z | Y_i P(j)X_iX_j|z \rangle =-\cos^{d_i-1}(2\theta)\sin(2\theta).
\end{align*}
Here, the second last equality follows since $\langle z_k|P(k)|z_k\rangle=0$ for all $k\neq i,j$. A similar calculation shows that
\begin{align*}
\langle \phi | Y_iY_j | \phi \rangle &= -\cos^{d_j-1}(2\theta)\sin(2\theta).
\end{align*}
Finally, for the last term in Eq.~\eqref{eq:3terms}, we will need to take the triangles into account: 
\begin{align}
&\langle \phi | Z_iZ_j | \phi \rangle \\&= \langle z |\bigg(\prod_{\{k,\ell\} \in E}\exp(-i\theta P(k) P(\ell))\bigg)Z_iZ_j\bigg(\prod_{\{k,\ell\} \in E}\exp (i\theta P(k) P(\ell))\bigg)|z \rangle\nonumber\\
&= \langle z |\bigg(\prod_{k:k\neq i,\{k,j\} \in E}\exp(-2i\theta P(k)X_j)\prod_{k:k\neq j,\{i,k\} \in E}\exp(-2i\theta Y_iP(k))\bigg)Z_iZ_j|z \rangle\nonumber\\
&= -\langle z | \prod_{k: k \neq i,\{k,j\} \in E} (\cos(2\theta) - i\sin(2\theta) P(k) X_j)\prod_{k \neq j:\{i,k\} \in E} (\cos(2\theta) - i\sin(2\theta) Y_i P(k))|z \rangle\label{eq:prodterms}\\
&= -\sum_{a = 0}^{\lfloor {T \over 2} \rfloor} \binom{T}{2a} \sin^{4a}(2\theta)\cos^{d_i + d_j - 2 - 4a}(2\theta))\label{eq:evensum}
\end{align}
Where the last equality is obtained by noting that a pair of triangles $\{i,j,k\}$ and $\{i,j,l\}$ will give a term $(-i\sin(2\theta))^4Y_iP(k)Y_iP(l)X_jP(k)X_jP(l) = \sin^4(2\theta)\cdot I$ inside the expectation value in Eq.~\eqref{eq:prodterms}. The summation in Eq.~\eqref{eq:evensum} runs over all even cardinality subsets of triangles. Thus, when $z_i\neq z_j$, 
\begin{align*}
&\langle \phi (\theta) | 2h_{ij} | \phi (\theta) \rangle =\\
& 1 + \sin(2\theta)\cos^{d_i-1}(2\theta) + \sin(2\theta)\cos^{d_j-1}(2\theta)  +
\sum_{a = 0}^{\lfloor {T \over 2} \rfloor} \binom{T}{2a} \sin^{4a}(2\theta)\cos^{d_i + d_j - 2 - 4a}(2\theta)\\
&= 1 + \sin(2\theta)\cos^{d_i-1}(2\theta) + \sin(2\theta)\cos^{d_j-1}(2\theta)+
\frac{1}{2}\left(\frac{1}{\cos^{2T}(2\theta)} + \frac{\cos^{T}(4\theta)}{\cos^{2T}(2\theta)}\right)\cos^{d_i+d_j-2}(2\theta)\\
&= 1 + \sin(2\theta)\cos^{d_i-1}(2\theta) + \sin(2\theta)\cos^{d_j-1}(2\theta)+
 \frac{1 + \cos^{T}(4\theta)}{2}\cos^{d_i+d_j-2-2T}(2\theta).
\end{align*} 
Here, the second equality follows from the binomial expansion of $\frac{1}{2}\left((1+x)^T+(1-x)^T\right)$ for $x=\frac{\sin^2(2\theta)}{\cos^2(2\theta)}$.

\paragraph{Unsatisfied edge:} Suppose $z_i=z_j=0$. Then $P(i)=Y_i$ and $P(j)=Y_j$ and so 
\begin{align*}
\langle \phi | Y_iY_j | \phi \rangle &=  \langle z | Y_iY_j|z \rangle=0.
\end{align*}
On the other hand, 
\begin{align}
&\langle \phi | X_iX_j | \phi \rangle= \langle z |\bigg(\prod_{\{k,\ell\} \in E}\exp(- i\theta P(k) P(\ell))\bigg)X_iX_j\bigg(\prod_{\{k,\ell\} \in E}\exp(i\theta P(k) P(\ell))\bigg)|z \rangle\nonumber\\
&= \langle z |\bigg(\prod_{k: k \neq j,\{i,k\} \in E} \exp(-2i\theta Y_i P(k))\bigg)\bigg(\prod_{k: k \neq i, \{k,j\} \in E}\exp(- 2i\theta P(k) Y_j)\bigg)X_iX_j|z \rangle\nonumber\\
&= \langle z |\prod_{k: k \neq j,\{i,k\} \in E} (\cos(2\theta) - i\sin(2\theta) Y_i P(k))\prod_{k: k \neq i, \{k,j\} \in E} (\cos(2\theta) - i\sin(2\theta) P(k) Y_j)X_iX_j|z \rangle\nonumber\\
&= \sum_{a = 1}^{\lfloor {T + 1 \over 2} \rfloor} \binom{T}{2a-1} \sin^{4a-2}(2\theta)\cos^{d_i + d_j - 4a}(2\theta).\label{eq:oddtriangles}.
\end{align}
The summation in Eq.~\eqref{eq:oddtriangles} runs over all odd cardinality subsets of triangles. Finally,
\begin{align*}
&\langle \phi | Z_iZ_j | \phi \rangle \\
&= \langle z | \prod_{k: k \neq j,\{i,k\} \in E} (\cos(2\theta) - i\sin(2\theta) Y_i P(k))\prod_{k: k \neq i,\{k,j\} \in E} (\cos(2\theta) - i\sin(2\theta) P_k Y_j)Z_iZ_j|z \rangle\\
&= \sum_{a = 0}^{\lfloor {T \over 2} \rfloor} \binom{T}{2a} \sin^{4a}(2\theta)\cos^{d_i + d_j - 2 - 4a}(2\theta)\\
\end{align*}
Where again, we have used the fact that any pair of triangles will result in an identity term. 

If $z_i=z_j=1$, similar calculations show that the contributions from $\langle \phi | Y_iY_j | \phi \rangle$ and $\langle \phi | X_iX_j | \phi \rangle$ are interchanged, but that their sum is unchanged. So for any unsatisfied edge, we have added lines to this equation:
\begin{align*}
\langle \phi | 2h_{ij} | \phi \rangle &= 1 -\sum_{a = 1}^{\lfloor {T + 1 \over 2} \rfloor} \binom{T}{2a-1} \sin^{4a-2}(2\theta)\cos^{d_i + d_j - 4a}(2\theta) - \sum_{a = 0}^{\lfloor {T \over 2} \rfloor} \binom{T}{2a} \sin^{4a}(2\theta)\cos^{d_i + d_j - 2 - 4a}(2\theta)\\
&=1-\cos^{d_i + d_j - 2}(2\theta)\left(\sum_{b = 0}^{T} \binom{T}{b} \frac{\sin^{2b}(2\theta)}{\cos^{2b}(2\theta)}\right)=1-\cos^{d_i + d_j - 2-2T}(2\theta),
\end{align*}
where we used the binomial expansion of $(1+x)^T$ for $x=\frac{\sin^2(2\theta)}{\cos^2(2\theta)}$.
\end{proof}

\section{Acknowledgments}
AA is supported by the Canadian Institute for Advanced Research, through funding provided to the Institute for Quantum Computing by the Government of Canada and the Province of Ontario. Perimeter Institute is also supported in part by the Government of Canada and the Province of Ontario. DG acknowledges the support of the Natural Sciences and Engineering Research Council of Canada, IBM Research, and the Canadian Institute for Advanced Research. We thank Sevag Gharibian, Eunou Lee, and Ojas Parekh for comments and helpful discussions. KM is grateful for support from Vanier CGS and the TQT Quantum Graduate Visitors Program. 
\bibliographystyle{plain}
\bibliography{bibliog}
 \end{document}